\documentclass{article}
\usepackage{amssymb} 
\usepackage{amsmath}
\usepackage{amsfonts}
\usepackage{natbib}
\usepackage{amsthm}   
\usepackage{epsfig}
\usepackage{multirow}

\newfont{\handw}{cmmi10 scaled 1200}

\newtheorem{Prop}{Proposition}[section]
\newtheorem{Lem}[Prop]{Lemma}
\newtheorem{LemDef}[Prop]{Lemma and Definition}
\newtheorem{Th}[Prop]{Theorem}

\newtheorem{Rm}[Prop]{Remark}
\newtheorem{Def}[Prop]{Definition}

\newtheorem{Cor}[Prop]{Corollary} 
 
\newfont{\smcal}{cmu10 scaled 1200}

\newcommand{\cov}{\operatorname {COV}}
\newcommand{\re}{\operatorname {Re}}
\newcommand{\im}{\operatorname {Im}}
\newcommand{\grad}{\operatorname {grad}}
\newcommand{\diag}{\operatorname {diag}}

\newcommand{\st}{\operatorname {such~that~}}

\newcommand{\tr}{\operatorname {tr}}

\newcommand{\rank}{\operatorname {rank}}

\newcommand{\argmin}{\operatorname {argmin}}

\begin{document}
	\title{Inference on 3D Procrustes Means: \\Tree Bole Growth, Rank Deficient Diffusion Tensors and Perturbation Models 
	\\ {\small Dedicated to the Memory of Herbert Ziezold (1942 -- 2008)}}
   \author{Stephan Huckemann \footnote{Supported by Deutsche Forschungsgemeinschaft Grant MU 1230/10-1  and DFG Graduate School 1023. Part of this research is contained in the habilitation thesis of the author.}}
\date{}
    \maketitle

\begin{abstract} 


	The Central Limit Theorem (CLT) for extrinsic and intrinsic means on manifolds is extended to a generalization of Fr\'echet means. Examples are the Procrustes mean for 3D Kendall shapes as well as a mean introduced by Ziezold. This allows for one-sample tests previously not possible, and to numerically assess the `inconsistency of the Procrustes mean' for a perturbation model and `inconsistency' within a model recently proposed for diffusion tensor imaging. Also it is shown that the CLT can be extended to mildly rank deficient diffusion tensors. An application to forestry gives the temporal evolution of Douglas fir tree stems tending strongly towards cylinders at early ages and tending away with increased competition.  

\end{abstract}
\par
\vspace{9pt}
\par
\vspace{9pt}Figure 
\noindent {\it AMS 2000 Subject Classification:} \begin{minipage}[t]{6cm}
Primary 62H11\\ Secondary 62G20, 62H15
 \end{minipage}
\par

\section{Introduction}\label{intro-scn}

	The study of descriptors of geometrical objects encompassed by terms of \emph{form} or \emph{shape} be it for artisanry, biological and morphological interest or medical applications dates back to the very origins of mankind. It took, and 
	it still takes, however, until modern days to fully realize, that underlying most seemingly intuitive concepts of shape are rather counter-intuitive non-Euclidean geometries. Adding to counter-intuition,  going from 2D to 3D shapes, these geometries cease to be well behaved. For this reason, 
	the statistical analysis of 3D shapes is highly challenging and has gained much less attention in theory and practice. 

	This work has been motivated by a joint research with the Institute for Forest Biometry and Informatics at the University of G\"ottingen studying the temporal evolution of the 3D shape of tree boles over their growing periods. The task tackled here 
	is to assess their growth towards and away from cylinders over time. This research is not only of interest in understanding fundamentals of biological growth and subsequent model building, cf. \cite{SH98,KoiHir06}. Also, the deviation from cylindricity 
	has a direct economical impact by reducing log volume and increasing the number of turns to reposition logs for commercial sawing processes, cf. 
	\cite{RBWE07}.


	While Procrustes analysis is a well established tool for the statistical analysis of shape (e.g. \cite{DM98} and \cite{Dshapes} for numerical routines), to the knowledge of the author, for 3D there are no asymptotic results available. Rather there is the belief that Procrustes sample means are unqualified for inference even in 2D because they may be ``inconsistent estimators'' of the ``shape of the mean'' -- the latter being the mean of a \emph{perturbation model} -- unless the error is isotropically distributed, see \cite{Lele93}, \cite{KM97} as well as \cite{Le98}. This perturbation model introduced by \cite{G91} and  
	discussed in detail in Section \ref{pert_mods:scn} (cf. (\ref{perturbation_model:eq}) on page \pageref{perturbation_model:eq}), assumes error on the original configurations. The results cited above allowed inference based on Procrustes means only for a limited set of 2D scenarios, e.g. when randomness is caused by a procedure of acquiring landmarks independent of rotation and location. For more general settings as often occur in applications of biology, e.g. when randomness is also the result of non-uniform biological growth, however, Procrustes means could only contribute to data description while they were not considered for use in inferential statistics.

	Such more general applications in mind, \cite{BP03,BP05} established an asymptotic theory for \emph{intrinsic} and \emph{extrinsic} means on manifolds allowing for 
	\emph{non-parametric} asymptotic inference. Here  ``non-parametric'' stands for the general approach not silently relying on a mean given by a perturbation model in the configuration space. Curiously in 2D, Procrustes means coincide with extrinsic means, the latter being ``consistent'' in the statistical sense i.e. satisfying a Strong Law of Large Numbers,  
	see \cite{Z77}, \cite{Le98} as well as \cite{BP03}.  
	Due to the dominating paradigm of the perturbation model, it seems 
	that the notion of a Procrustes \emph{population} mean has not quite found its way into the community. 

	Recently, a perturbation model has also motivated a non-standard approach in the context of neuro-imaging by \cite{DKZL09}, employing Procrustes analysis for estimating a mean diffusion tensor. 
	Although ``perturbation consistency'' has not been established, 
	practical imaging results in particular for nearly degenerate tensors have been of very good quality.

	It is the intention of this work to
	\begin{enumerate}
	 \item[(a)] make 3D Procrustes means with their well behaved asymptotics available for the practioners in the field,
	 \item[(b)] to provide for an assessement of ``inconsistency'' with perturbation models  on Kendall's shape spaces as well as for the approach of \cite{DKZ09} for diffusion tensor imaging, and
	\item[(c)] to derive a framework to obtain the temporal shape motion relative to the shape of cylinders for a sample of Douglas fir tree stems. 
	\end{enumerate}

	To this end, the following Section \ref{GPA:scn} reviews Procrustes analysis and related results. Section \ref{SLLN_CLT_PZ_mean:scn} places Procrustes means and a mean introduced by \cite{Z94} in the context of Fr\'echet means and gives the Central-Limit-Theorem (CLT) as the main theoretical result. Here, two aspects require special attention: first, a CLT can only hold on manifolds  while in 3D, the underlying shape spaces cease to be that; secondly, since 3D Procrustes means are neither intrinsic nor extrinsic means, the CLT of \cite{BP05} is accordingly modified as detailed in the appendix.

	In Section \ref{pert_mods:scn} it is clarified 
	that the reported ``(in)consistency'' of Procrustes means expresses  ``(in)compatibility'' of a perturbation model with the canonical shape space geometry. Moreover for 3D, simulations based on the CLT show 
	that for isotropic errors the perturbation model 
	can  be considered nearly compatible in most practical situations unless the shape of the perturbation mean is almost degenerate. 
	In contrast by simulation in Section \ref{midly:scn}, a similar perturbation model for mean diffusion tensors 
	is not compatible, not even for small isotropic errors. Curiously, however, with increased degeneracy of the perturbation mean considered, incompatibility seems not increasing. This is also the case for an extension of the model to \emph{mildly rank deficient} diffusion tensors, presented here.

	Section \ref{Conical_boles:scn} introduces a framework for the assessment of shape of frusta cut from tree boles. It turns out that the shapes of cylinders form a geodesic in the shape space allowing for the computation of the distance of arbitrary frusta to the space of cylinders. 
	For statistical inference, Boostrap confidence intervals for this distance are simulated from a sample of reconstructed tree ring structures of small size. Comparison with two typical growth scenarios gives that 
	the bole shape of young trees with little competition grows uniformly or stronger towards cylinders, while the motion of shapes of old tree boles with heavy competition away from cylinders is similar to a growth maintaining cross-sectional ellipticality and tapering.

\section{Procrustes Analysis and Kendall's Shape Spaces}\label{GPA:scn}
	
	In the statistical analysis of similarity shapes based on landmark configurations, geometrical $m$-dimensional objects (usually $m=2,3$) are studied by placing $k>m$ \emph{landmarks} at specific locations of each object. Each object is then described by a matrix in the space $M(m,k)$ of $m\times k$ matrices, each of the $k$ columns denoting an $m$-dimensional landmark vector. $\langle x,y\rangle := \tr(xy^T)$ denotes the usual inner product with norm $\|x\| = \sqrt{\langle x,x\rangle}$. For convenience and without loss of generality for the considerations below, only \emph{centered} configurations are considered. Centering in way treating all landmarks equally can be achieved by multiplying with a sub-Helmert matrix ${\cal H}$

	\begin{eqnarray}\label{Helmert_matrix:def}{\cal H} &=& \left(\begin{array}{cccc}
		   \frac{1}{\sqrt{2}} &\frac{1}{\sqrt{6}} & \dots& \frac{1}{\sqrt{k(k-1)}} \\ 
	   -\frac{1}{\sqrt{2}} &\frac{1}{\sqrt{6}} & \dots& \frac{1}{\sqrt{k(k-1)}} \\ 
	   0&-\frac{2}{\sqrt{6}} &\dots& \frac{1}{\sqrt{k(k-1)}}\\ 
	   \vdots&\vdots&\ddots&\vdots \\
	   0&0 &\dots& -\frac{k-1}{\sqrt{k(k-1)}}
	            \end{array}\right) \in M(k,k-1)
	\end{eqnarray}
		from the right, yielding $x{\cal H}$ in $M(m,k-1)$. This corresponds to the relocating of objects in space such that their mean landmark is zero and isometrically projecting the linear sub space of matrices with vanishing mean landmark to the space of matrices with $k-1$ of landmarks. For this method and other centering methods cf. \citet[Chapter 2]{DM98}. Excluding also all matrices with all landmarks coinciding gives the space of \emph{configurations} 
	\begin{eqnarray*}
	F_m^k&:=& M(m,k-1) \setminus \{0\} \,.
	\end{eqnarray*}

	Since only the similarity shape is of concern, all configurations are considered modulo the group of similarity transformations $H = SO(m) \times \mathbb R_+$ (recall that we excluded w.l.o.g. translations) where the action is given by the 
	operation $(g,\lambda)x= g\lambda x.$
	 Here, $SO(m)$ denotes the special orthogonal group (the orientation preserving orthogonal transformations). The  \emph{shape} of a configuration $x\in F_m^k$ is the orbit $[x] = \{hx:h\in H\}$, the \emph{shape space} is the quotient $F_m^k/H = \{[x]:x\in F_m^k\}$ with a suitable metric.

	\paragraph{A na\"ive shape space.}
	A straightforward metric structure for $F_m^k/H$ is given by the canonical quotient metric:
	$$\inf_{h_1,h_2\in H} \|h_1x_1 -h_2x_2\|\mbox{ for }x_1,x_2\in F_m^k\,.$$
	Unfortunately, due to the scaling action of $\mathbb R_+$, this metric is identically zero, a dead end for statistical ambition.

	\paragraph{General Procrustes analysis (GPA).}
	 As a workaround, \cite{Gow} introduced a constraining condition which allowed for the definition of a mean. For a sample of configurations $x_1,\ldots,x_n\in F_m^k$ a \emph{full Procrustes sample mean} is given by the shape of
	 $$\frac{1}{n}\sum_{j=1}^n h^*_jx_j$$
	where $h^*_1,\ldots,h^*_n$ are minimizers over $h_1,\ldots, h_n \in H$ of the \emph{Procrustes sum of squares} 
	\begin{eqnarray*}
	 \sum_{i,j=1}^n\|h_jx_j -h_ix_i\|^2 \mbox{ under the constraining condition } \left\|\frac{1}{n}\sum_{j=1}^n h_jx_j\right\| ~=~ 1\,.
	\end{eqnarray*}
	Letting $\mu = \frac{1}{n}\sum_{j=1}^n h_jx_j$, $h_j = (g_j,\lambda_j)$, the Procrustes sum of squares is 
	$$ \sum_{i,j=1}^n\|h_jx_j -h_ix_i\|^2 = 2n\sum_{j=1}^n\|h_jx_j -\mu\|^2 = 2n\sum_{j=1}^n\|\lambda_jg_jx_j -\mu\|^2\,.$$ 
	W.l.o.g. we may assume that all configurations are contained in the unit sphere
	$$ S_m^k :=\{x\in M(m,k-1): \|x\|=1\}\,.$$ 
	Then, any minimizing $\mu$ is the orthogonal projection of the mean of minimizing $h_1x_1,\ldots, h_nx_n$ to $S_m^k$. In consequence, minimization can be performed sequentially, first for the $\lambda_j$, then for the $g_j$ and finally for $\mu$, as noted. Partial differentiation in particular gives the minimizing $\lambda^*_j = \langle g_jx_j,\mu\rangle >0$ (unless $\langle gx_j,\mu\rangle = 0 $ for all $g\in SO(m)$) such that every $\mu^*$ having the shape of a  full Procrustes sample mean is a minimizer of
	$$ \min_{\mu\in S_m^k} \min_{\footnotesize\begin{array}{l}g_1,\ldots,g_n \in SO(m)\\\langle g_ix_i,y\rangle \geq 0,i=1,\ldots,n\end{array}} \sum_{j=1}^n \| \langle g_jx_j,\mu\rangle \,g_jx_j -\mu\|^2 = \min_{\mu\in S_m^k} \sum_{j=1}^n \delta(x_j,\mu)^2$$
	with the \emph{residual distance}
	\begin{eqnarray}\label{res_dist:def}
	\delta(x,y) &:=& \min_{\footnotesize\begin{array}{l}g\in SO(m)\\\langle gx,y\rangle \geq 0\end{array}}\|\langle gx,y\rangle gx - y\|\,.
	\end{eqnarray}

	\paragraph{Kendall's shape spaces.}
	\cite{K77} proposed a slightly different work\-around. Instead of considering the quotient w.r.t. the action of $\mathbb R_+$, he projected all configurations to $S_m^k$ called the \emph{pre-shape sphere} and considered the metric quotient w.r.t. $SO(m)$ only, which is called \emph{Kendall's shape space}
	$$\Sigma_m^k := S_m^k/SO(m) = \{[x]:x\in S_m^k\}\mbox{ with the \emph{fiber} } [x] = \{gx:g\in SO(m)\}\,.$$
	Kendall's shape space is thus a quotient of a sphere allowing for two non-trivial canonical quotient metrics
	$$d_{\Sigma_m^k}([x],[y]) := \inf_{g,h\in SO(m)} d(gx,hy) = \inf_{g\in SO(m)} d(gx,y)\,,$$
	$d$ denoting either the Euclidean distance on $F_m^k$ or the \emph{spherical} distance on $S_m^k$:
	$$d^{(e)}_{F_m^k}(x,y) :=\|x-y\|,\quad d^{(s)}_{S_m^k}(x,y) :=2\arcsin\frac{\|x -y\|}{2}\,.$$ 

	In some applications only the \emph{form}, i.e. shape without size filtered out is of interest. The corresponding \emph{size-and-shape space} is $S\Sigma_m^k := F_m^k/SO(m)$ with \emph{size-and-shape}
	$[x] := \{gx: g\in SO(m)\}\in S\Sigma_m^k $ of $x\in  F_m^k$.
	For a sample $x_1,\ldots, x_n \in F_m^k$, \emph{partial Procrustes sample means}, the size-and-shapes of $\mu^*\in F_m^k$ are then considered which minimize

	$$\min_{g_1,\ldots,g_n\in SO(m)} \sum_{j=1}^n\|\mu - g_j x_j\|$$
	over $\mu \in F_m^k$. Obviously, $\mu^* =\frac{1}{n}\sum_{j=1}^n g_j^*x_j$ with minimizers $g_1^*,\ldots,g_n^* \in SO(m)$. These means are ``partial'' because no equivalence w.r.t. scaling is considered. In contrast, in the definition of full Procrustes means above equivalence under the full group of similarities is considered. 

	\paragraph{2D full Procrustes means are extrinsic means.}
	For $m=2$, \cite{K84} observed that one may
	identify $F_2^k$ with $\mathbb C^{k-1}\setminus \{0\}$ such that every landmark column corresponds to a complex number. This means in particular that $z\in \mathbb C^{k-1}$ is a complex row-vector. With the Hermitian conjugate $a^* = (\overline{a_{kj}})$ of a complex matrix $a=(a_{jk})$ the pre-shape sphere $S_2^k$ is identified with $\{z\in \mathbb C^{k-1}: zz^*=1\}$ on which $SO(2)$ identified with $S^1=\{\lambda \in\mathbb C: |\lambda|=1\}$ acts by complex scalar multiplication. Then the well known Hopf-Fibration gives $\Sigma_2^k=S_2^k/S^1=\mathbb CP^{k-2}$, the complex projective $(k-2)$-dimensional space. Moreover, denoting by $M(k-1,k-1,\mathbb C)$ all complex $(k-1)\times  (k-1)$ matrices, the \emph{Veronese-Whitney embedding} is given by
	\begin{eqnarray*}
	 \frak{v}:\Sigma_2^k &\to& \{a \in M(k-1,k-1,\mathbb C): a^*=a\}
	,~~[z] ~\mapsto~ z^*z\,.
	\end{eqnarray*}
	If $Z\in \mathbb C^{k-1}$ is a random pre-shape, identifying  $\frak{v}(\Sigma_2^k)$ with $\Sigma_2^k$, the set of \emph{extrinsic means} (cf.  \cite{BP03}) of $[Z]$ is the set of shapes of the orthogonal projection of the usual expected value $\mathbb E(Z^*Z)$ to $\frak{v}(\Sigma^k_2)$. Employing complex linear algebra, the set of extrinsic means is easily identified as the 
	shapes of the eigenvectors to the largest eigenvalue of the \emph{complex integral of squares matrix} $\mathbb E(Z^*Z)$. 
	Since on the other hand
	$$\mathbb E\big(\delta(Z,w)^2\big) = 1 - w \mathbb E(Z^*Z) w^*\,.$$
	with the residual shape distance from (\ref{res_dist:def}), we have that that any eigenvector of  $\mathbb E(Z^*Z)$ to its largest eigenvalue is a pre-shape of a full Procrustes mean. 

	\begin{Th} The set of 2D full Procrustes means is the inverse image under the Veronese-Whitney embedding of the set of extrinsic means on  $\frak{v}(\Sigma_2^k)$. 
	\end{Th}
 
	In most applications, the largest eigenvalue of $\mathbb E(Z^*Z)$ is simple, then the 2D full Procrustes mean is uniquely determined.

	A similar procedure gives extrinsic means using the \emph{Schoenberg embedding} for the related space of Kendall's \emph{reflection shapes} of arbitrary dimension, not discussed here, cf. \citet{B08}.  
	Procrustes means for three- and higher-dimensional configurations, however, cannot be modeled in this vein. 

\section{Asymptotics for Procrustes and Ziezold Means}\label{SLLN_CLT_PZ_mean:scn}

	In this section we will see that 3D and higher-dimensional Procrustes means as well as a mean introduced by \cite{Z94} are special cases of Fr\'echet $\rho$-means. In the following, \emph{smooth} means at least twice continuously differentiable. 

\subsection{Kendall's Higher-dimensional Shape Spaces}
	As we have seen above, $\Sigma_2^k$ is a manifold, namely  a complex projective space. Similarly, $S\Sigma_2^m$ can be given a manifold structure. For $m\geq 3$, however, $\Sigma_m^k$ and $S\Sigma_m^k$ cease to be manifolds, they only contain dense and open submanifolds 
	{\footnotesize $$\begin{array}{rclcrcl}
	(\Sigma_m^k)^*&:=&(S_m^k)^*/SO(m)&\mbox{with}&(S_m^k)^*&:=&\{x\in S_m^k: \rank(x)\geq m-1\}\\
	(S\Sigma_m^k)^*&:=&(F_m^k)^*/SO(m)&\mbox{with}&(F_m^k)^*&:=&\{x\in F_m^k: \rank(x)\geq m-1\}	
	  \end{array}\,$$}called the \emph{manifold parts} of \emph{regular} or equivalently \emph{non-degenerate} shapes, size-and-shapes, pre-shapes and configurations, respectively. A pre-shape or configuration $x$ or its shape or its size-and-shape $[x]$ is called \emph{strictly regular} if $\rank(x)=m$. In case of $\Sigma_m^k$, $m\geq 3$, approaching degenerate shapes, some sectional curvatures are unbound. A detailed discussion can be found in \cite{KBCL99} 
	as well as \cite{HHM07}.

	Note that the square of a distance may not be smooth at so called cut points. E.g the square of the spherical distance between $e^{it}$ and $e^{is}$ on the unit circle $S^1\subset \mathbb C$ is smooth except when $e^{i(t-s)} =-1$. In general on a Riemannian manifold, the \emph{cut locus} $C(p)$ of $p$ comprises all points $q$ such that the extension of a length  minimizing geodesic joining $p$ with $q$ is no longer minimizing beyond $q$. If $q\in C(p)$ or $p\in C(q)$ then $p$ and $q$ are \emph{cut points}. 



	\begin{LemDef}\label{all_are_metrics:rm}
	All of
	{\footnotesize $$\begin{array}{rclcl}
	d^{(i)}_{\Sigma_m^k}([x],[y]) &:=& \min_{g\in SO(m)}d^{(s)}_{S_m^k}(gx,y)\,, \\
	d^{(p)}_{\Sigma_m^k}([x],[y]) &:=& \sin d^{(i)}_{\Sigma_m^k}([x],[y])&=&\min_{{\footnotesize\begin{array}{l}g\in SO(m)\\\langle gx,y\rangle \geq 0\end{array}}}\sqrt{1-\langle gx,y\rangle^2} \,,\\
	d^{(z)}_{\Sigma_m^k}([x],[y]) &:=& 2\sin \frac{d^{(i)}_{\Sigma_m^k}([x],[y])}{2}&=&\min_{g\in SO(m)}\sqrt{2(1-\langle gx,y\rangle)} \,,\\
	d^{(i)}_{S\Sigma_m^k}([x],[y]) &:=& \min_{g\in SO(m)}d^{(e)}_{F_m^k}(gx,y)&=&\min_{g\in SO(m)}\sqrt{\|x\|^2+\|y\|^2-2\langle gx,y\rangle} 
	\end{array}$$}\noindent
	are metrics on the shape space and the size-and-shape space, respectively, and their squares are smooth on the respective manifold parts except at cut points. 
	\end{LemDef}

	\begin{proof} If the canonical quotient $Q=M/G$ of a smooth Riemannian manifold $M$ due to the action of a Lie group $G=SO(m)$ on $M$ is a Riemannian manifold, then the intrinsic distance on $Q$ is a smooth metric (cf. \citet[Chapter 4.1]{AM78}). This gives the smoothness of $(d^{(i)}_{\cdot})^2$ on the respective manifold parts which are dense in the shape space and the size-and-shape space, respectively, except at cut points. 
	Hence, the $d^{(i)}_{\cdot}$ extend to metrics on these. Moreover, $d^{(z)}_{\Sigma_m^k}$ is a metric due to  convexity and monotonicity of $t \to 2\sin (t/2)$ for $t\in [0,\pi]$; for $d^{(p)}_{\Sigma_m^k}$ we rely on \citet[p. 206]{KBCL99}.
	\end{proof}

	We say that $x$ \emph{is in optimal position to $y$} if $\max_{g\in SO(m)} \langle gx,y\rangle = \langle x,y\rangle$. In particular then $d^{(z)}_{\Sigma_m^k}([x],[y]) =\|x-y\|$ and $d^{(i)}_{S\Sigma_m^k}([x],[y]) =\|x-y\|$, respectively.

\subsection{Fr\'echet $\rho$-Means}\label{Frech_means:scn}
	The Central-Limit-Theorem (CLT) derived below relies on the ``$\delta$-method'' for smooth transformations. In consequence, we can only expect a CLT theorem to hold on the manifold part of the shape space and the size-and-shape space, respectively. For the following we assume that $(Q,d)$ is a metric space and $P$ is a $D$-dimensional smooth Riemannian manifold. For example, $Q=\Sigma_m^k$ and $P = (\Sigma_m^k)^*$. Suppose that $X, X_1,X_2,\ldots$ are i.i.d. random variables mapping from an abstract probability space $(\Omega,\cal A,\mathbb P)$ to $Q$ equipped with its Borel $\sigma$-field. For a random vector $Y\in \mathbb R^D$, $\mathbb E(Y)$ denotes the usual expectation, if defined. 

	\begin{Def}\label{Frechet_means:def} For a continuous function $\rho:Q\times P \to [0,\infty)$  define the \emph{set of  population Fr\'echet $\rho$-means of $X$ in $P$} by
	$$ E^{(\rho)}(X) = \argmin_{\mu\in P} \mathbb E\big(\rho(X,\mu)^2\big) 
	\,.$$
	For $\omega\in \Omega$ denote the \emph{set of sample Fr\'echet $\rho$-means} by
	$$ E^{(\rho)}_n(\omega) = \argmin_{\mu\in P} \sum_{j=1}^n \rho\big(X_j(\omega),\mu\big)^2\,.$$
	\end{Def}

	Since their original definition  by \cite{F48} for $P=Q$ 
	and $\rho=d$, such means have found much interest. \cite{Z77} extended the concept to $\rho$ being a quasi-metric only. 
	On Riemannian manifolds ($Q=P$) w.r.t. the Riemannian metric $\rho$, \cite{BP03,BP05} introduced the corresponding means as \emph{intrinsic means}, and, taking $\rho$ to be the metric of an ambient Euclidean space, as \emph{extrinsic means}.

	In consequence of Lemma and Def. \ref{all_are_metrics:rm} and the connection with the residual distance (\ref{res_dist:def}), 
	$d^{(p)}_{\Sigma_m^k}([x],[y])=\delta(x,y)$, we can at once identify Procrustes means. 

	\begin{Cor} The set of $\rho$-Fr\'echet sample means on the shape space and size-and-shape space, respectively, are the sets of 
	\begin{enumerate}
	\item[(i)] full Procrustes sample means for $\rho =  d^{(p)}_{\Sigma_m^k}$,  
	\item[(ii)] partial Procrustes sample means for $\rho =  d^{(i)}_{S\Sigma_m^k}$.
	\end{enumerate}
	\end{Cor}

	For $m=2$, \cite{Z94} introduced mean shapes w.r.t. $\rho =  d^{(z)}_{\Sigma_m^k}$ which have been studied also by \cite{Le98}. The following first two definitions honor his memory. The other extend sample Procrustes means to the population case. 

	\begin{Def}
	 Call $d^{(z)}_{\Sigma_m^k}$ the \emph{Ziezold metric} on the shape space, and the Fr\'echet $d^{(z)}_{\Sigma_m^k}$-means the set of \emph{Ziezold means}. The Fr\'echet $d^{(p)}_{\Sigma_m^k}$-means are the set of \emph{full Procrustes means}, the Fr\'echet $d^{(i)}_{S\Sigma_m^k}$-means the set of \emph{partial Procrustes means}.
	\end{Def}
	Ziezold means and full Procrustes means, even though defined earlier, can be thought of as a generalization of extrinsic means and means for crude residuals (cf. \cite{MJ00}), respectively. Partial Procrustes means are extensions of intrinsic means to non-manifolds. 

	Let us now touch on the issues of existence and uniqueness for the above introduced means. Both are are rather simple issues for extrinsic means, since they are orthogonal projections of classical Euclidean means in an ambient space to the embedded manifold in question, this projection being well defined except for a set of Lebesgue measure zero (cf. \cite{BP03}). For non-extrinsic means as above, however, existence, namely that the means are assumed on the manifold part is a deeper issue tackled in \cite{H_meansmeans_10}. Intrinsic means are unique if the underlying random elements are sufficiently concentrated (cf. \cite{L01} for the elaborate proof). For higher-dimensional Ziezold and full Procrustes means, to the knowledge of the author, there are no uniqueness results available. 
%

	Since we are concerned with metrics, 
	the following is a consequence of the Strong Law by \cite{Z77}, cf. also \cite{BP03}.

	\begin{Th}\label{SLLN:th} Suppose that $X$ is a random pre-shape in $S_m^k$ or a random configuration in $F_m^k$ with unique Ziezold, full Procrustes or partial Procrustes mean $[\mu]$. 
	Then every measurable selection $[\mu_n(\omega)]$ from the sets of sample Ziezold, full Procrustes or partial Procrustes means, respectively, is a strongly consistent estimator of $[\mu]$. 
	\end{Th}

	For the following we need the condition that $[X]$ is bounded away from the cut locus of 
	the mean a.s. in the Ziezold or full/partial Procrustes 
	sense 
	\begin{eqnarray}\label{bdd_away:cond}
	\exists \epsilon >0 &\st& d([X],[\mu])<
	d(C([\mu]),[\mu])-\epsilon \mbox{ a.s.}
	\end{eqnarray}
	where $[\mu]$ and $d$ are corresponding  Ziezold, full Procrustes or partial Procrustes means and distances, respectively.

\subsection{The Central-Limit-Theorem}\label{CLT:scn}
	The Central-Limit-Theorem for Fr\'echet $\rho$-means requires additional setup. 

		\begin{Def}\label{CLT:def} Let $P$ be a $D$-dimensional manifold. We say that a $P$-valued estimator $\mu_n(\omega)$ of $\mu \in P$ satisfies a \emph{Central-Limit-Theorem} (CLT), if in any local chart $(\phi,U)$ near $\mu =  \phi^{-1}(0)$ there are a suitable $D\times D$ matrix $A_{\phi}$ and a Gaussian $D\times D$ matrix ${\cal G}_{\phi}$ with zero mean and semi-definite symmetric covariance matrix $\Sigma_{\phi}$ such that
	$$ \sqrt{n}A_{\phi}\big(\phi(\mu_n) - \phi(\mu)\big)~\to~{\cal G}_{\phi}$$
	in distribution as $n\to \infty$. 
	\end{Def}
	In most applications $A_{\phi}$ is non-singular, then in consequence of the ``$\delta$-method'', for any other chart $(\phi',U)$ near $\mu =  \phi'^{-1}(0)$ we have simply
		 $$A^{-1}_{\phi'}\Sigma_{\phi'} (A^{-1}_{\phi'})^T = J(\phi'\circ\phi^{-1})_0 A^{-1}_{\phi} \Sigma_{\phi}(A^{-1}_{\phi})^TJ(\phi'\circ\phi^{-1})_0^T\,$$
	where $J(\cdot)_0$ denotes the Jacobian matrix of first derivatives at the origin.

	In consequence of Theorem \ref{SLLN:th} and Theorem \ref{CLT:th} from the appendix (assertion (ii) follows from \cite{BP05}) we have the following. 
	\begin{Th}\label{CLT_PZ:th}
	Suppose that $X$ is a random pre-shape in $S_m^k$ or a random configuration in $F_m^k$, $[X]$ having a unique Ziezold, full Procrustes or partial Procrustes mean $[\mu]$ on the respective manifold part.
	Then every measurable selection $[\mu_n(\omega)]$ from the sets of sample Ziezold, full Procrustes or partial Procrustes means, respectively, satisfies a Central-Limit-Theorem if $[X]$ is bounded away from the cut locus 
	of the mean a.s. in the Ziezold or full/partial Procrustes sense (\ref{bdd_away:cond}) under the following additional condition:
	\begin{enumerate}
	\item[(i)] none in case of Ziezold or full Procrustes means, 
	\item[(ii)]  in case of partial Procrustes means, if the Euclidean second moment $\mathbb E(\|X\|^2)$ is finite.
	\end{enumerate}\noindent
	In a suitable chart $(\phi,U)$ the corresponding matrices from Definition \ref{CLT:def} are given by
	$$ A_{\phi} = \mathbb E(H\rho([X],[\mu])),~~\Sigma_{\phi} = \cov(\grad\rho([X],[\mu]))$$
	where $\rho$ denotes the distances $d_{\Sigma_m^k}^{(p)},d_{\Sigma_m^k}^{(z)}$ and $d_{S\Sigma_m^k}^{(i)}$, respectively. Moreover, $\grad$ and $H$ denote the gradient and Hessian of $x\mapsto \rho([X],[\phi^{-1}(x)])^2$, respectively. 
	\end{Th}

	Numerical simulations show a rather good finite sample accuracy of the CLT, cf. Figure \ref{geod_pert_1_sample_test:fig}.
	\begin{figure}[h!]
	\begin{minipage}{0.4\textwidth}
	\includegraphics[angle=-90,width=1\textwidth]{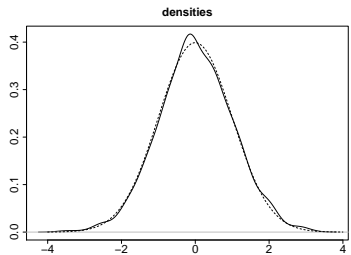}
	\end{minipage}
	\begin{minipage}{0.05\textwidth}\hfill \end{minipage}
	\begin{minipage}{0.55\textwidth}
	\caption{\it Depicting the proximity to the standard normal density (dashed line) of the normalized density (solid line) of the non-zero coordinate of full Procrustes means of samples of size $n=10$ from the geodesic perturbation model (\ref{geod_perturbation_model:eq}) in $\Sigma_3^4$ with $\epsilon_i=0$ and $s_i \sim N(0,0.1^2)$ with the two-dimensional mean $\mu=\diag(1,-1,0)/\sqrt{2}{\cal H}^T$ and  $\nu=\diag(1,1,1)/\sqrt{3}{\cal H}^T$.\label{geod_pert_1_sample_test:fig}}
	\end{minipage}
	\end{figure}
	\begin{Rm}
	 For $m\geq 3$, computing Ziezold means (for an algorithm cf. \cite{Z94}) is computationally less costly than computing full Procrustes means (corresponding 
	algorithms 
	are discussed in \citet[Chapter 5.3]{DM98}): for full Procrustes means in every iteration step every single optimal positioned datum needs additionally to be projected to the tangent space. The simulations reported in Section \ref{isotropic_error:scn} and the Bootstrap simulations in Section \ref{Conical_boles:scn} 
	give similar results but are approximately $15 \,\%$ faster when using Ziezold means instead of full Procrustes means. 
	\end{Rm}

	\begin{Rm}
	 \cite{Gr05} proves that the above algorithms converge under rather broad conditions and supplies error estimates. 
	\end{Rm}

\section{Perturbation Inconsistency for Kendall Shapes}\label{pert_mods:scn}
		For brevity in this section unless otherwise specified, Procrustes means refer to full Procrustes means.

	\cite{G91} proposed to model a sample of landmark configuration matrices $y_i$ ($i=1,\ldots,n$) with a \emph{perturbation model} that since then, has been highly popular in the community:
	\begin{eqnarray}\label{perturbation_model:eq}
	 y_i &=& \lambda_i g_i(\mu + \epsilon_i) + t_i\,.
	\end{eqnarray}
	The $\lambda_i >0$ convey scaling, the $g_i$ are rotation matrices and the $t_i$ stand for translations. Obviously, these three sets of nuisance parameters keep the \emph{shape} invariant, of interest is only the deterministic \emph{perturbation mean} $\mu$ and the random perturbations $\epsilon_i$ with zero expectation. If the size is also of interest then set $\lambda_i=1$ keeping only  the form invariant under rotation and translation. Under the assumption of isotropic Gaussian errors $\epsilon_i$, 
	for estimation and inference on the population perturbation means $\mu_i$ and error covariances $\epsilon_i$, \cite{G91} proposed to use \emph{sample Procrustes means} obtained from  GPA, cf. Section \ref{GPA:scn}. In the sequel, the property that (for arbitrary errors) sample Procrustes means converge to the shape of the mean of an underlying perturbation model has been coined as the ``consistency of Procrustes means''.

\subsection{Isotropic 3D Error}\label{isotropic_error:scn}

	In order to validate Goodall's proposal, \cite{KM97} studied the analog of the perturbation model (\ref{perturbation_model:eq}) on the pre-shape sphere 
	and showed that for 2D configurations with isotropic $\epsilon_i$, the Procrustes population mean from (\ref{perturbation_model:eq}) is identical with the shape of the perturbation mean.
	 \cite{Le98} extended these results and showed that under slightly relaxed conditions for 2D configurations, intrinsic, Ziezold and Procrustes means all agree with the shape of $\mu$. 


	For 3D shapes the above arguments are no longer valid because the shape-fibres in $S_m^k$ are not spanned by geodesics in general (see \citet[Example 5.1]{HHM07}), hence equality of Procrustes means with the shape of a perturbation model cannot be expected, even for very small isotropic error. 

	In order to assess the practical impact of this effect, we measure the distance between the shape of the perturbation mean and its corresponding Procrustes mean. In view of Theorem \ref{CLT_PZ:th}, this can be done by determining the distance between the shape of the perturbation mean and a corresponding sampled Procrustes mean while confidence or equivalently its accuracy can be estimated by distances of correspondingly sampled sample Procrustes  means to their sample Procrustes mean. The following considerations detail this setup.

	Consider a random $X = Y{\cal H}/\|Y{\cal H}\| \in S_m^k$ with $Y$ from a perturbation model (\ref{perturbation_model:eq}) such that the following means are unique. Assume $\mu \in S_m^k$ and denote by 
	\begin{description}
	\item[$\nu$] the pre-shape of the Procrustes population mean 
	in optimal position to $\mu$; 
	\item[$\hat{\nu}^{(n)}$] the random pre-shape in optimal position to $\nu$ of the Procrustes sample mean 
	of $X_1,\ldots,X_n$ i.i.d. as $X$; 
	\item[$\hat{\nu}^{(n)}_j$] for $j=1,\ldots,N$ i.i.d. realizations of $\hat{\nu}^{(n)}$, 
	\item[$\hat{\nu}^{(n,N)}$] a realization of the Procrustes sample mean of $\hat{\nu}^{(n)}_1,\ldots,\hat{\nu}^{(n)}_N$ in optimal position to $\nu$. 
	\end{description}
	Then in consequence of Theorem \ref{CLT_PZ:th},
	\begin{eqnarray*}
	\hat{\sigma}^{(n,N)}&:=&
	 \sqrt{\frac{1}{N}\sum_{j=1}^Nd^{(p)}_{\Sigma_m^k}([\hat{\nu}^{(n)}_j], [\hat{\nu}^{(n,N)}])^2} 
	~\approx~ \sqrt{\mathbb E\big(d^{(p)}_{\Sigma_m^k}([\hat{\nu}^{(n)}],[\nu])^2\big)}
	\end{eqnarray*}
	gives an approximation of the confidence into or accuracy of the measurement of $\nu$ by $\hat{\nu}^{(n)}$. 
	On the other hand, we have the approximation
	$$\hat{d}^{(n,N)}:=\sqrt{\frac{1}{N}\sum_{j=1}^Nd^{(p)}_{\Sigma_m^k}([\nu_j^{(n)}],[\mu])^2}\approx  \sqrt{\mathbb E\big(d^{(p)}_{\Sigma_m^k}(\nu^{(n)},\mu)^2\big)}\,.$$
	The goodness of both approximations depends on $N$.
	Moreover for fixed $N$, $\hat{\sigma}^{(n,N)}$ can be made arbitrarily small by choosing $n$ sufficiently large. If alongside, $\hat{d}^{(n,N)}$ also becomes equally small, this can be taken as strong evidence that the shape of the perturbation model agrees with the corresponding Procrustes mean. Otherwise, we have strong evidence, that the two disagree. Table \ref{pert_mod_3d:tab} reports values of  $\hat{d}^{(n,N)}$ and $\hat{\sigma}^{(n,N)}$ for typical simulations of (\ref{perturbation_model:eq}) for $m=3,k=4$ with isotropic i.i.d Gaussian error $\epsilon$ of variance $\sigma^2$. The following remark summarizes the results of these and further simulations.

	\begin{table}\centering
 	\begin{tabular}{c|c|cc}$\mu {\cal H}$&$\sigma$&$\hat{d}^{(n,N)}$&$\hat{\sigma}^{(n,N)}$\\\hline
	\multirow{3}{*}{$\frac{1}{\sqrt{3}}\left(\begin{array}{ccc}1&0&0\\0&1&0\\0&0&1\end{array}\right)$}
	&$0.5$&$ 0.011$&$ 0.0097$\\
 	 &$0.1$&$0.0023 $&$ 0.0021$\\
	&$0.01 $&$0.00022 $&$0.00020$\\ \hline
	\multirow{3}{*}{$\frac{1}{\sqrt{1.1}}\left(\begin{array}{ccc}1&0&0\\0&0.3&0\\0&0&0.1\end{array}\right)$}
	&$0.5$&$0.30$&$ 0.030$\\
	&$0.1$&$0.016$&$ 0.0022$\\
	&$0.01$&$ 0.00026$&$ 0.00020$\\ \hline
	\multirow{3}{*}{$\frac{1}{\sqrt{1.001}}\left(\begin{array}{ccc}1&0&0\\0&0.01&0\\0&0&0\end{array}\right)$}&$0.1 $&$ 0.12$&$ 0.082$\\
	&$0.01$&$0.0055$&$ 0.00044$\\
	&$0.001$&$5.0e-05$&$ 2.0-05$
 	\end{tabular}
	\caption{\it Estimated distance $\hat{d}^{(n,N)}$ between shape of the mean $\mu$ of the perturbation model (\ref{perturbation_model:eq}) and corresponding Procrustes mean with standard error $\hat{\sigma}^{(n,N)}$ for $n=10,000$ and $N=10$. The error follows independently ${\cal N}(0,\sigma^2)$ in every component. For convenience the Helmertized pre-shape $\mu {\cal H}$ is reported displaying proximity to degeneracy.\label{pert_mod_3d:tab}}
	\end{table}

	\begin{Rm}
	 In approximation the perturbation model (\ref{perturbation_model:eq}) with isotropic errors is compatible with the geometry of $\Sigma_3^k$ if the perturbation mean is far from degenerate. For nearly degenerate perturbation means, however, the perturbation model may be incompatible even for small isotropic errors. Using Ziezold means instead of Procrustes means gives similar results. 
	\end{Rm}

\subsection{Non-Isotropic 2D Error}\label{Kents_reserv:scn}

	\paragraph{Global effects.}
	Recall  from Section \ref{GPA:scn}, modeling 
	$\Sigma_2^k$ by a complex projective space. Since the eigenvectors of the complex integral of squares matrix
	may vary discontinuously under continuous matrix variation this led \cite{KM97} to the conclusion that Procrustes means can be inconsistent where in fact they are statistically consistent but possibly discontinuous. 

	\begin{figure}[h!]
	\begin{minipage}{0.6\textwidth}
	 \includegraphics[angle=-90,width=1\textwidth]{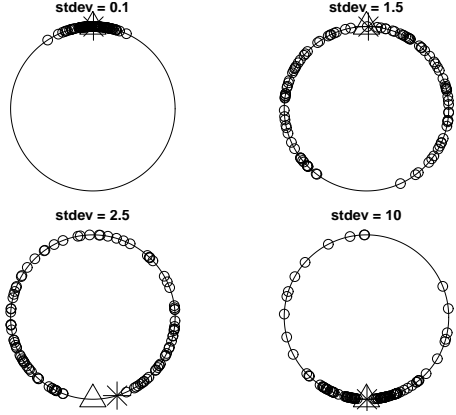}
	\end{minipage}
	\begin{minipage}{0.35\textwidth}
	\caption{\it Great circle in $\Sigma_2^3$ spanned by the shapes of the perturbation (\ref{Kent_pert:mod}). The little circles denote $100$ sampled shapes, the star their Procrustes sample mean and the triangle the respective Procrustes population mean. The respective headers give the standard deviation of the Gaussian $\epsilon$. 
	\label{Kents_Reservation:fig}}
	\end{minipage}
	\end{figure}

	In view of the perturbation model (\ref{perturbation_model:eq}) consider in the spirit of \citet[p. 285]{KM97} the complex configuration $\mu = (\sqrt{2},0,1/\sqrt{2})\in \mathbb C^3$ additively perturbed by the non-isotropic random $\epsilon = (0,0,\sqrt{6}\eta t /2)\in \mathbb C^3$, $t\sim {\cal N}(0,1)$, $\eta>0$, modeling linear configurations with two fixed endpoints and a random landmark varying in the middle. The Kendall pre-shape of 
	\begin{eqnarray}\label{Kent_pert:mod}
	\mu + \epsilon
	\end{eqnarray} 
	is given by $z = \frac{1}{\sqrt{1+\eta^2t^2}}(1, \eta t)$ with the complex integral of squares matrix
	$$ \left(\begin{array}{cc} \mathbb E\left(\frac{1}{1+\eta^2t^2}\right)&0\\
	0& \eta^2 \mathbb E\left(\frac{t^2}{1+\eta^2t^2}\right)\end{array}\right)\,.$$
	For small error intensity $\eta$, the shape of $\mu$ gives the Procrustes population mean. For higher error intensity, however, the Procrustes population mean indeed changes abruptly to the shape of the configuration $\nu=(\sqrt{2},\sqrt{2},-2\sqrt{2})$. 

	In order to visualize the situation in the shape space recall the Hopf fibration $$S_2^3 \to \Sigma_2^3 = \mathbb CP^1: (\alpha_1,\alpha_2)\mapsto \left(\re(\alpha_1\overline{\alpha_2}),\im(\alpha_1\overline{\alpha_2}), \frac{|\alpha_1|^2-|\alpha_2|^2}{2} \right)$$ 
	mapping from a three-sphere to a two-sphere. The counter-intuitive behavior of discontinuity  of the Procrustes mean is visualized in Figure \ref{Kents_Reservation:fig}. As is clearly visible, the discontinuity of the Procrustes mean in the error's standard deviation $\sqrt{6}\eta/2$ is due to the fact, that the corresponding shapes move along a common great circle in $\Sigma_2^3$, from clustering at the top (the shape of $\mu$) to clustering at the bottom (the shape of $\nu$). The discontinuity observed by \cite{KM97} thus reflects a global effect of an ``incompatibility'' of the perturbation model with the geometry of the shape space.

\paragraph{Local effects.}

	As above, 
	\citet[Figure 3 on p. 598]{Lele93} considers an example of a distribution of planar but now quadrangular configurations along a straight line in the configuration space $F_2^4$. For simplicity of the argument consider $\mu = (i,1,-i,-1),~~\epsilon = (0,i,0,-i)$ and the sample $z_1=\mu +\epsilon,~~z_2=\mu-\epsilon\,$, $\epsilon >0$,
	with mean $\mu$ in Euclidean $F_2^4$. One verifies immediately that $z_1$ and $z_2$ are not in optimal position, rather $\lambda = \frac{1}{\sqrt{5}}(1+2i)$ puts $z_1$ into optimal position   $\lambda z_1$ to $z_2$ w.r.t. to the action of $SO(2)$. Then the form of $(\lambda z_1 + z_2)/2 $ is the partial Procrustes mean, which is different from the form of $\mu$. The alleged ``inconsistency'' of the partial Procrustes mean is now a local effect of an ``incompatibility'' of the perturbation model with the canonical geometry of the size-and-shape space $S\Sigma_2^4$.

\section{Asymptotics for Diffusion Tensors}\label{midly:scn}

	In diffusion tensor neuro-imaging (for a short introduction into this young field, e.g. \cite{VilZhKL06}), the dominating eigenvector of a symmetric semi-positive definite (SPD) matrix from the space 
	$P(m) = \{0\neq a\in M(m,m): a^T=a\geq 0\}$ 
	exhibits the dominating direction of molecular displacement due to a flow in tissue fibres of interest. The statistical and non-statistical literature to the task of reconstructing SPD matrices which are called \emph{diffusion tensors} (DTs) in this context is vast. Recently, matching techniques involving shape analysis (e.g. \cite{CMMWY06}) have gained momentum. 

	A non-standard statistical approach to reconstruct a mean DT from an observed sample of neighboring DTs based on Procrustes methods has been proposed by \cite{DKZL09}. One motivation comes from perturbation models, which, as it turns out below are ``inconsistent'' with Procrustes means. Coming as a surprise, however, practical image reconstructions based on this method in particular for nearly degenerate DTs, have produced results of convincingly good quality, cf. \cite{DKZ09}. 
	In the following, this approach is first extended to \emph{mildly rank-deficient DTs} and secondly, perturbation inconsistency is assessed.
%
%
	Even though in most applications to imaging, the flow is observed in 2D or 3D, i.e. $m=2,3$, the following 
	applies to arbitrary $m\in \mathbb N$. 

	Diffusion tensors observed in the ``real world'' fall into the sub-space 
	$P^{+}(m)=\{a\in P(m): a>0\}$ 
	of symmetric strictly positive definite matrices. This space --  even though much more complicated than complex projective space -- has been well studied (as the ``universal symmetric space of non-compact type with non-positive non-constant sectional curvatures'' e.g. \citet[Chapter XII]{L99}). Via the well known \emph{Cholesky factorization} 
	$$ {\rm chol} : P^{+}(m) \to UT^{+}(m)\,,$$ 
	this space can equivalently be modelled by the group of $m$-dimensional upper triangular matrices $UT^{+}(m)$ with positive diagonal. 
	Traditionally diffusion tensors have been modelled within these spaces, statistical analysis has either been carried out by using the extrinsic distance inherited from $M(m,m)$ (e.g. \cite{PajBass03}) or the intrinsic distance due to the aforementioned structure (cf. \cite{fletch4}). Modeling flow in ideal micro-fibres, obviously leads to rank-deficient diffusion tensors, which, however, are not contained in these manifolds. Unless one utilizes a Euclidean embedding e.g. in the space of symmetric matrices, a new structure has to be found. 

\subsection{A CLT for Mildly Rank-deficient Diffusion Tensors}
	Here, an embedding $P(m)\hookrightarrow \Sigma_{m}^{m+1}$ is proposed, that maps the space $$P^*(m) = \{a\in P(m): \rank(a)\geq m-1\}$$
	 of \emph{mildly rank-deficient diffusion tensors} into the manifold part $(\Sigma_{m}^{m+1})^*$. This embedding is inspired by recent work of \cite{DKZ09} who, in effect, model $P^{+}(m)$ as a subspace of $\Sigma_{m}^{m+1}$. More subtly one could model as well using Kendall's \emph{reflection size-and-shape space}. In fact, modeling within size-and-shape space is equivalent to embedding reflection size-and-shape space in size-and-shape space.

	To this end consider the following canonical domain for upper triangular matrices
	{\footnotesize\begin{eqnarray*}
	 \lefteqn{UT(m) }\\&=& \Big\{a\in M(m,m): \mbox{ there is }1\leq i_0 \leq m\mbox{ such that for } \\ 
	&&\left.\begin{array}{rl}
	i=1:& \mbox{there is }m\geq j_{1} \geq 1 \mbox{ with }a_{ij_1} >0\mbox{ and } a_{il} = 0 \mbox{ for all }1\leq l < j_1,\\
	2\leq i \leq i_0: &
	\mbox{there is }m\geq j_{i} \geq j_{i-1}+1\mbox{ with } a_{il} = 0 \mbox{ for all }1\leq l<j_i,\mbox{ and }\\& \hspace{3.4cm}a_{ij_i} > 0\mbox{ in case } i < m,\\
	i_0<i:  
	&a_{ij}=0\mbox{ for all } j=1,\ldots,m\end{array}
	\right\}\,,
	\end{eqnarray*}}the 
	sphere $SUT(m) = \{a\in UT(m): \|a\|=1\}$, $UT^\geq(m) = \{a\in UT(m): a_{mm}\geq 0\}$ and  $SUT^\geq(m) = SUT(m)\cap UT^\geq(m)$. Moreover, with a bijective extension of the Cholesky factorization, the corresponding canonical projections from Section \ref{GPA:scn}:
	$\pi$ and $s\pi$, $s(x) = \|x\|$ and bijections $\phi, \phi_s$, consider the following diagram of mappings:

	\vspace{-0.2cm}
	{\footnotesize\begin{eqnarray}\label{wild_P(m)_structure:def}
	\left.\begin{array}{rclclcl}
	 P(m) \\
	\updownarrow {\rm chol}\\
	UT^\geq(m)&\hookrightarrow& SO(m)\times UT(m)&\stackrel{(g,a) \mapsto ga}{\rightarrow}&F_m^{m+1}&\stackrel{(\pi,s)}{\to}&\Sigma_m^{m+1}\times \mathbb R_+\\
	&&&&\downarrow s\pi&& \updownarrow \phi \\
	&&&&S\Sigma_{m}^{m+1}&& SUT^\geq(m)\\
	&&&&\updownarrow \phi_s\\
	&&&&UT^\geq(m)
	\end{array}\right\}
	\end{eqnarray}}In
	particular, (\ref{wild_P(m)_structure:def}) gives rise to the following mappings
	$$\tau_s:P(m) \to S\Sigma_m^{m+1}\mbox{ and } \tau:P(m) \to \Sigma_m^{m+1}\,.$$


	\begin{Th}\label{diff_tens_kend_ss:thm} The mappings $\tau_s$ and $\tau$ are well defined, open  and continuous. Moreover
	 \begin{enumerate}\item[(i)] $\tau_s$ restricted to the space of mildly rank-deficient diffusion tensors $P^*(m)$ is an injective mapping into the manifold part $(S\Sigma_m^{m+1})^*$ of Kendall's size-and-shape space. Restricted to the space of full rank diffusion tensors it is a diffeomorphism onto an open subset. 
	\item[(ii)] $\tau$ restricted to the space of mildly rank-deficient diffusion tensors $P^*(m)$ maps into the manifold part $(\Sigma_m^{m+1})^*$. Restricted to the space of full rank diffusion tensors it is a submersion of codimension $1$ onto an open subset. 
	\end{enumerate}
	\end{Th}

	For the proof of Theorem \ref{diff_tens_kend_ss:thm} and 
	diagram (\ref{wild_P(m)_structure:def}) we refer to the appendix.

	\begin{Rm}
	In consequence, inference for mildly deficient-rank diffusion tensors can be carried out via $\tau_s$ or $\tau$ utilizing the CLT: Theorem \ref{CLT_PZ:th}. Moreover, mean shapes can be pulled back under $(\tau_s)^{-1}$ to obtain \emph{mean diffusion tensors} if 
	they 
	stay away from the region corresponding to $a_{mm} <0$ in $UT(m)$. 
	\end{Rm}

\subsection{Perturbation Inconsistency for Diffusion Tensors}


	A typical perturbation model considered by \cite{DKZ09} 
	is 
	\begin{eqnarray}\label{cov_matrix_pert:mod}x_i &=& (\mu + \epsilon_i)^T(\mu +\epsilon_i)\end{eqnarray}
	with perturbation mean $\mu \in UT(m)$ and error $\epsilon_i$ following a Gaussian  distribution ${\cal N}(0,\sigma^2)$ 
	independently in every component or in every upper diagonal component and zero in every strictly lower diagonal component. In order to relate the quality of the embedding $\tau_s$ from (\ref{wild_P(m)_structure:def}) to other structures for the space of diffusion tensors (e.g. that of the universal symmetric space), \cite{DKZ09} compare the partial Procrustes distance between $\tau_s(\mu^T\mu)$ and partial Procrustes means from samples of (\ref{cov_matrix_pert:mod}). 

	For illustration we report in Table \ref{pert_mod_tens_3d:tab} a simulation in analogy to Section \ref{isotropic_error:scn} (using partial Procrustes means instead of full Procrustes means, 
	cf. Table \ref{pert_mod_3d:tab}) for three typical diffusion tensors for 
	 upper triangular isotropic errors. The situation is similar for isotropic errors. The values of $\hat{d}^{(n,N)}$ for $\sigma=0.1$ in our Table \ref{pert_mod_tens_3d:tab} correspond to the ``RMSE($d_S$) of $\hat{\Sigma}_S$'' reported in the blocks labelled ``II'' in ``Table 2'' (corresponding to our second block) and ``Table 3'' (corresponding to our third block) of \cite{DKZ09}. Here, however, we used $n=10,000$ and identified these numbers as a measure only for perturbation inconsistency by additionally computing the standard error.
	We remark the consequence.

		\begin{table}[h!]\centering
 	\begin{tabular}{c|c|cc}$\mu$&$\sigma$&$\hat{d}^{(n,N)}$&$\hat{\sigma}^{(n,N)}$\\\hline
	\multirow{3}{*}{$\left(\begin{array}{ccc}1&0&0\\0&1&0\\0&0&1\end{array}\right)$}
	&$0.1$&$0.12$&$ 0.0017$\\
 	 &$0.01$&$0.012$&$ 0.00017$\\
	&$0.001 $&$0.0012$&$1.6e-05$\\ \hline
	\multirow{3}{*}{$\left(\begin{array}{ccc}1&0&0\\0&0.3&0\\0&0&0.1\end{array}\right)$}
	&$0.1$&$0.023$&$0.0020$\\
	&$0.01$&$0.00026$&$0.00023$\\
	&$0.001$&$2.4e-05$&$ 2.5e-05$\\ \hline
	\multirow{3}{*}{$\left(\begin{array}{ccc}1&0&0\\0&0.01&0\\0&0&0\end{array}\right)$}
	&$0.1 $&$ 0.14$&$ 0.0019$\\
	&$0.01$&$0.0085$&$ 0.00021$\\
	&$0.001$&$0.00081$&$ 2.2-05$
 	\end{tabular}
	\caption{\it Estimated distance $\hat{d}^{(n,N)}$ between size-and-shape of $\mu$ from the perturbation model (\ref{cov_matrix_pert:mod}) and corresponding partial Procrustes mean with standard error $\hat{\sigma}^{(n,N)}$ for $n=10,000$ and $N=10$. The error follows independently ${\cal N}(0,\sigma^2)$ in every upper triangular  component.\label{pert_mod_tens_3d:tab}}
	\end{table}

	\begin{Rm}\label{DTI-inconsistency:rm}
	Suppose that $X$ follows a perturbation model (\ref{cov_matrix_pert:mod}). For  upper triangular perturbation mean $\mu$ and error $\epsilon$ this is an anisotropic perturbation model for $S\Sigma_m^{m+1}$. For $\epsilon$ independent Gaussian in every component, $\mathbb E\big({\rm chol}(X)\big)\neq {\rm chol}(\mu)$ in general (cf. Lemma \ref{chol:lem} in the appendix). 
	 Simulations corroborate inconsistency 
	with 
	the geometry of $S\Sigma_3^k$. 
	Surprisingly it seems that incompatibility is not increasing, near degeneracy. This observation may be taken as an explanation for successful modeling of nearly rank deficient diffusion tensors via Kendall's size-and-shape spaces and certainly deserves further research.
	\end{Rm}




\section{Assessing Tree Bole Cylindricity 
	}\label{Conical_boles:scn}



	We conclude with an application from forest biometry.
	As basic descriptors for the shape of tree boles,  
	\emph{taper curves} 
	relate height above ground level with the area of a cross section at that height. 
	The shape of a tree stem is thus described by a 
	two-dimensional curve. Empirical curves can be directly analyzed with methods of shape analysis, e.g. \cite{Kr02} or, sophisticated models based on biological and elastomechanical context can be sought for, e.g. \cite{CS94} or  \cite{GaSlobMats98}. For a discussion of current taper curve models cf. \cite{LiWeis10}. 
	While taper curves, following Cavalieri's principle, model tree stems by \emph{circular frusta}, recently bole ellipticality has been studied more closely, e.g. \cite{SH98} or \cite{KoiHir06}. In particular, \cite{RBWE07} model 3D logs by \emph{elliptical frusta} and discuss the impact of their findings on commercial logging: proximity to circular frusta increases the volume produced and reduces the number of turns to reposition logs in sawing machines and thus the total sawing time.

	\begin{figure}[htb!]
	 \includegraphics[width=1\textwidth]{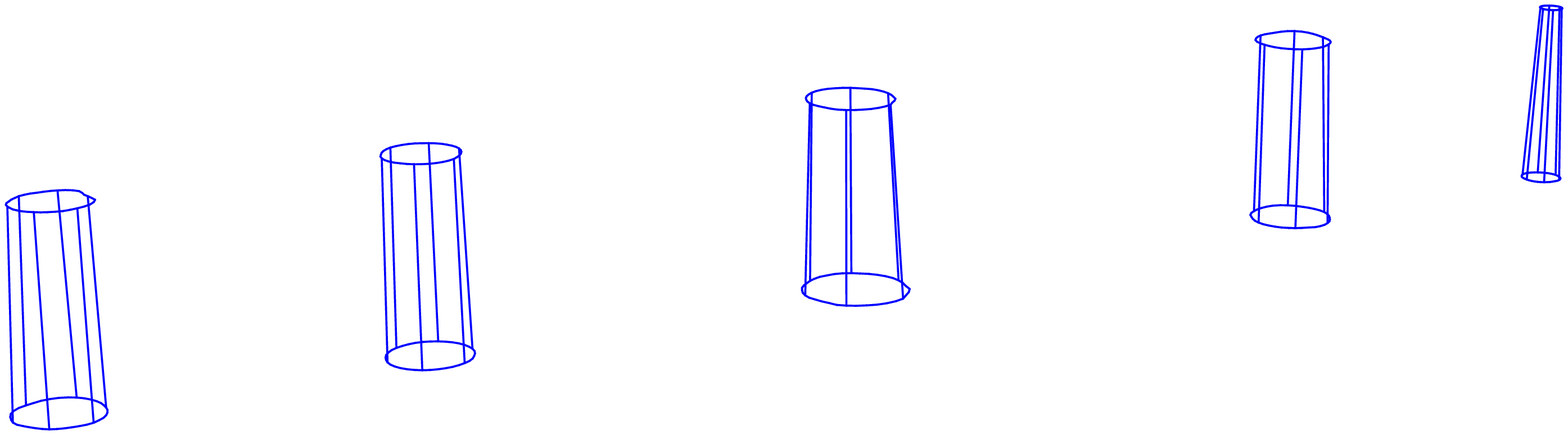}

	\vspace{-0.5cm}
	 \includegraphics[width=1\textwidth]{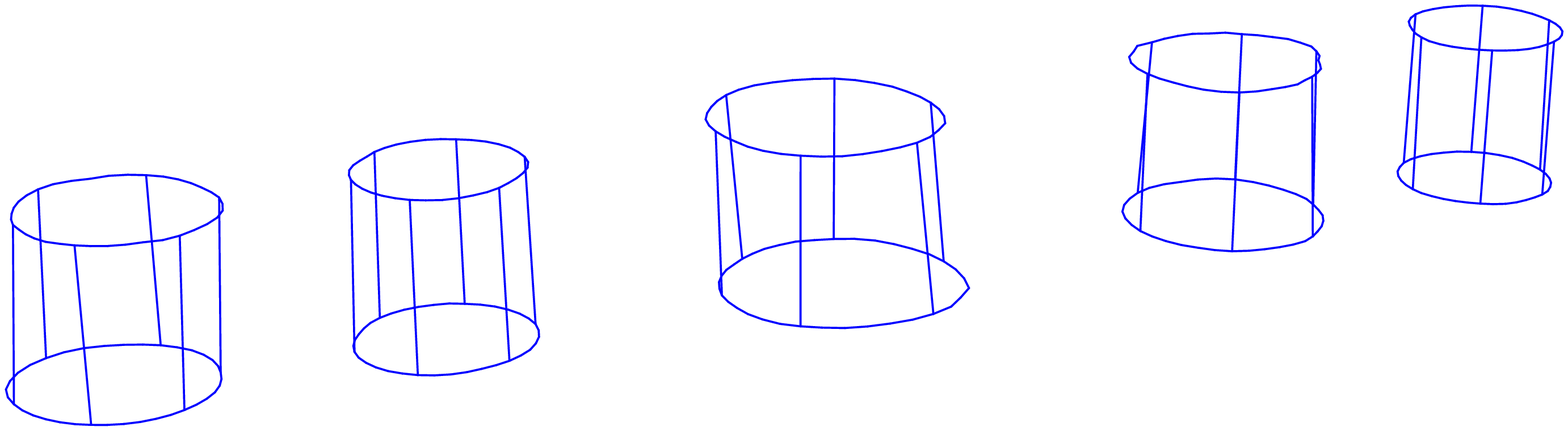}

	\vspace{-0.5cm}
	 \includegraphics[width=1\textwidth]{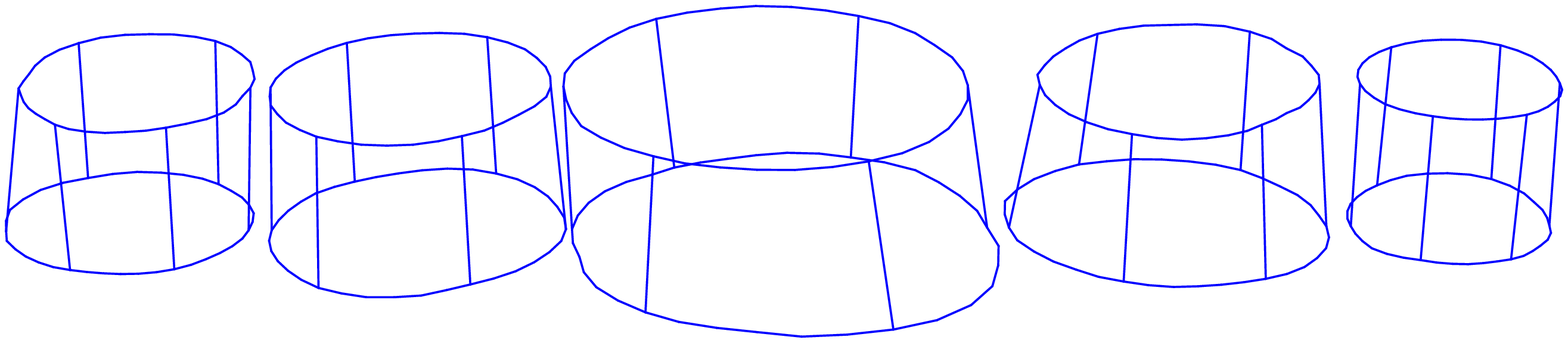}

	\vspace{-0.5cm}
	\caption{\it Distorted view (the center frustum in the bottom row has a diameter of approx. $50$ cm) along 36 angles of 1 m butt logs of five Douglas firs at ages 8 years (top row), 30 years (middle row) and 62 years (bottom row).\label{hex-cones:fig}}
	\end{figure}
\subsection{Data Description}
	Usually, tree stems are divided into three different parts: a short neiloid bottom part with strong tapering connecting with the root system, the main bole with little tapering -- which is of prime commercial interest -- and a conical top (e.g. \cite{LiWeis10}). Typically a \emph{butt log} that is a frustum taken above \emph{breast height} (1.3 m) is used to assess bole quality. For our application, we use 1 m butt logs from five Douglas fir trees typical for the inside of a small experimental stand in the Netherlands as detailed by \cite{gaslob01}. At about the age of 10 to 15 years, tree crowns met; subsequently with almost no thinning, competition for light increased strongly. The entire ring structure of bottom and top disk has been elaborately reconstructed along 36 equally spaced angles allowing to reconstruct the butt logs for every age beginning from 8 years to 62 years as displayed in Figure \ref{hex-cones:fig} for early, intermediate and ultimate age.
\subsection{Elliptical-Like Frusta: A Totally Geodesic Subspace}
	In order assess the deviation of a 1 m frustum from a cylinder (of unknown dimensions), we compute the shape distance to the space of all cylinders. To this end, and in order to compare with a model building on elliptical frusta, we introduce a suitable subspace of configurations and shapes.

	For given  ${\bf a}^T,{\bf b}^T,{\bf 1}^T=(1,\ldots,1)^T \in \mathbb R^{\kappa}$ with
	\begin{eqnarray}\label{parallel_cone_section:cond}
	 \langle {\bf a},{\bf b}\rangle =0 =\langle {\bf a},{\bf 1}\rangle = \langle {\bf b},{\bf 1}\rangle\,
	\end{eqnarray}
	and $r,\alpha,t>0$ let
	\begin{eqnarray*}
	w_{\alpha,\beta,r,t} &:= &\left(\begin{array}{cc}  r\,{\bf a}&r t \,{\bf a}\\
       r\alpha\,{\bf b}&r \beta t\, {\bf b}\\
	0\, {\bf 1}&\,{\bf 1}
	\end{array}\right)\, 
	\end{eqnarray*}
	be the configuration of an \emph{elliptical-like} unit height frustum. In case of ${\bf a} = \big(\cos(2\pi/\kappa), \ldots, \cos(2\pi\kappa/\kappa)\big)$ and ${\bf b} = \big(\sin(2\pi/\kappa), \ldots, \sin(2\pi\kappa/\kappa)\big)$ we have an \emph{elliptical} frustum of mean radius $r$, tapering $t$ with bottom half axes of length $r,r\alpha$ and top half axes of length $tr,tr\beta$. A \emph{straight frustum} has $t=1$, a \emph{circular frustum} has $\alpha = 1=\beta$ and a \emph{cylinder} is a straight circular frustum.  On the grounds of the findings of \cite{CS94}, that for a large middle part of the bole, there is little shape variation when moving upward, a possible torsion between top and bottom ellipse is neglected in this model.
	Denote by
	$$ PS_{{\bf a},{\bf b}} := \left\{ \frac{w}{\|w\|} : w \in P_{{\bf a},{\bf b}} \right\}\subset S_3^{2\kappa}$$
	the pre-shapes of all \emph{elliptical-like frusta} determined by ${\bf a},{\bf b}\in \mathbb R^\kappa$ satisfying (\ref{parallel_cone_section:cond}). Here  
	$$ P_{{\bf a},{\bf b}} := \{ w_{\alpha,\beta,r,t}{\cal H}: \alpha,\beta,r,t>0\}\subset F_3^{2\kappa}$$
	and  ${\cal H}$ denotes the Helmert sub-matrix from (\ref{Helmert_matrix:def}) in Section \ref{GPA:scn}. 

	Recall that a submanifold $P\subset M$ is \emph{totally geodesic} if for any two points in $P$ any minimal geodesic segment joining the two in $M$ is contained in $P$.

	\begin{Th}\label{parallel_cone_scn_geodesic:thm} Consider ${\bf a}^T,{\bf b}^T\in \mathbb R^{\kappa}$ satisfying (\ref{parallel_cone_section:cond}). Then the shapes of all elliptical-like frusta form 
	a totally geodesic submanifold of $(\Sigma_3^{2\kappa})^*$ with horizontal lift  $PS_{{\bf a},{\bf b}}$ to $S_3^{2\kappa}$.
	\end{Th}
	The proof of Theorem \ref{parallel_cone_scn_geodesic:thm} 
	is found in the appendix. We have at once:
	\begin{Cor}\label{ideal_circular_cylinders:col}  
	The shapes of cylinders form a segment on a geodesic in $(\Sigma_3^{2\kappa})^*$. 
	\end{Cor}

\subsection{Data Analysis}

	The distance of the shape of a frustum to the geodesic spanned by the shapes of cylinders is computed via a method proposed in \cite{HHM07}. Since for every age considered there are only five frusta, a $95 \,\% $ confidence band for the distance of the full Procrustes mean to the geodesic of cylinders has been computed by $200$ Bootstrap resamples. As clearly visible in Figure \ref{conf_dist_ideal_circ_geod:fig}, young frusta until the age of approx. 15 years tend to toward cylinders. At later ages they tend away again. The change point at approx. 15 years can be explained by the fact that at this time 
	competition for light began.
 
	Note that 
	uniform growth naturally decreases tapering and ellipticality. In order to compare the initial growth to uniform growth, three curves of elliptical frusta 
	starting with $t$ and $\alpha=\beta$ in the range of corresponding estimated butt log values as well as
	with size growth identical to the mean size growth are depicted in the left display of Figure \ref{conf_dist_ideal_circ_geod:fig} ($95 \,\%$ of estimated tapering coefficients $t$ observed are between  $0.864 $ and $0.986$; a crude estimate for $\alpha$ and $\beta$ gives $95 \,\%$ of observed values between $1.02$ and  $1.13$ (after rotating such that $\alpha,\beta\geq 1$)).
	This comparison yields that  observed initial growth towards cylinders appears even stronger than uniform.

	In order to assess the growth of the older logs, compare it to the growth of elliptical frusta keeping tapering and ellipticality constant while letting size grow with the mean size of the original data. Three of such curves are depicted in the right display of Figure \ref{conf_dist_ideal_circ_geod:fig}. The observed movement away from the geodesics of cylinders appears rather similar. 

	\begin{figure}[h!]
	 \includegraphics[angle=-90,width=0.5\textwidth]{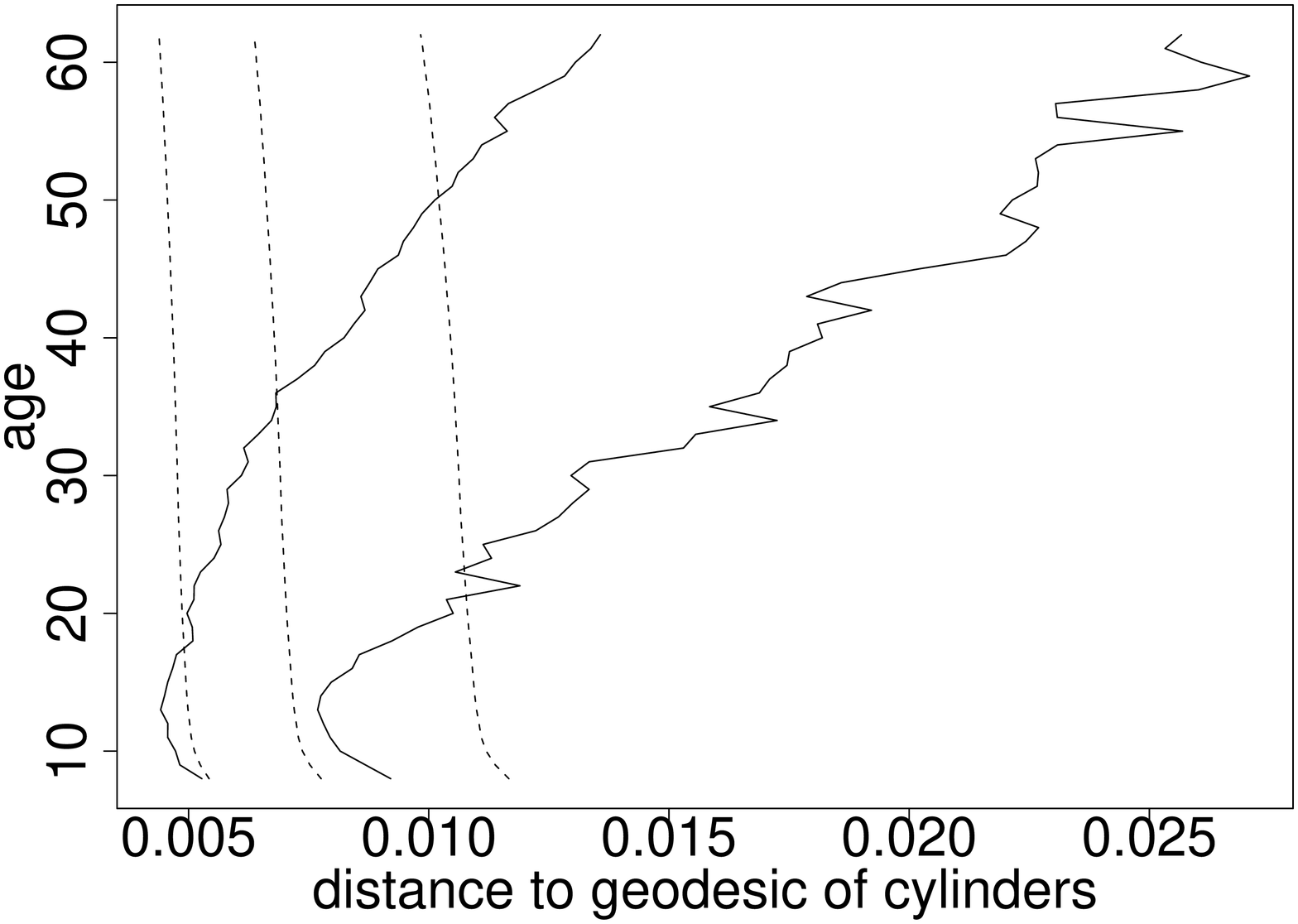}
	 \includegraphics[angle=-90,width=0.5\textwidth]{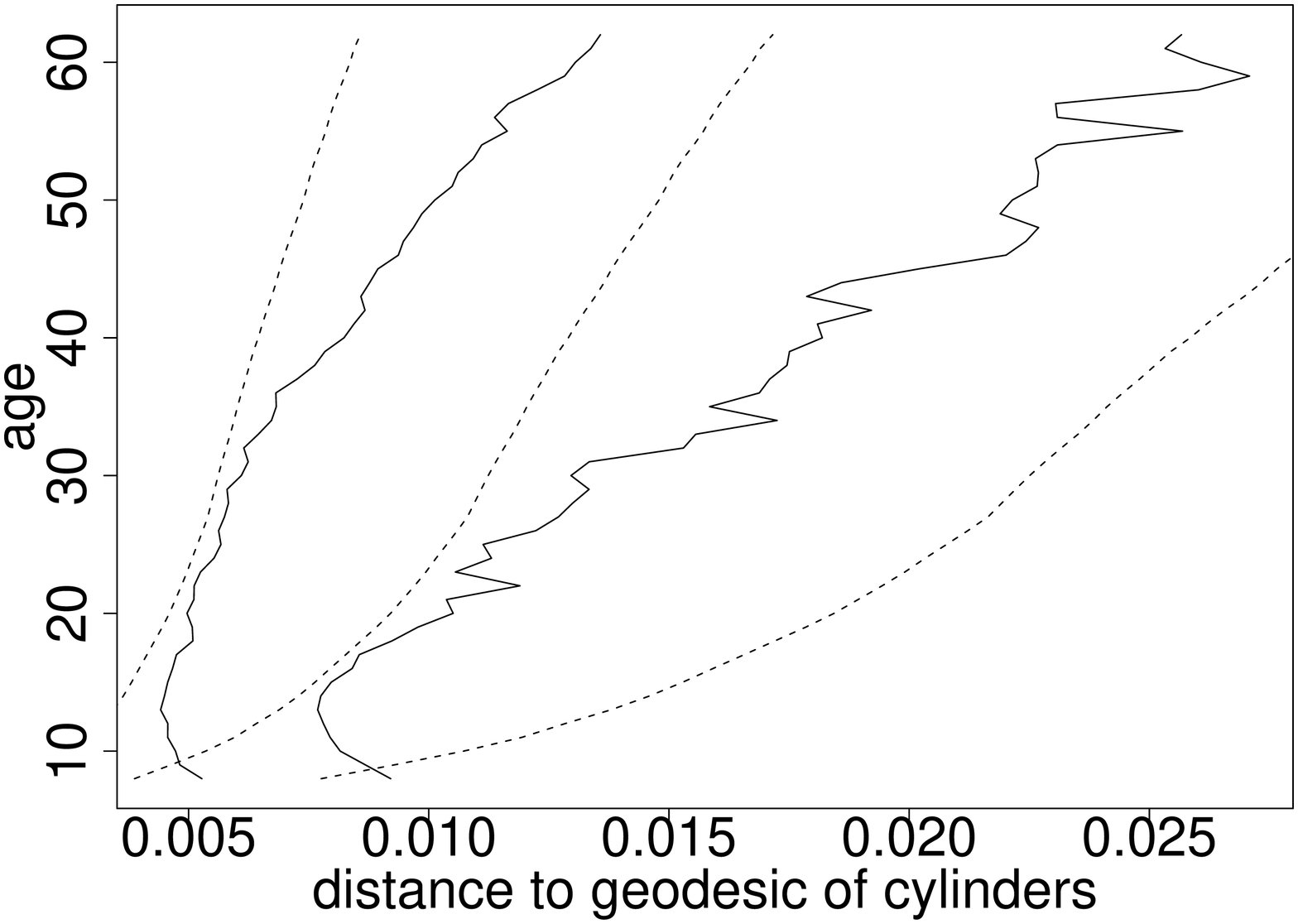}
	\caption{\it Solid: confidence band for the distance of the Procrustes mean shape of 1 m butt logs above breast height to the geodesic of cylinders for each year based on 200 boostrapped Procrustes means. For comparison dashed: uniform equalsize growth of elliptical cylinders (left display) starting with $(\alpha=\beta,t) = (1.07,0.93), (1.1,0.9), (1.15,0.85) $ from left to right; equalsize growth of elliptical cylinders (right display) with constant $(\alpha=\beta,t) = (1.025, 0.975), (1.05,0.95), (1.1,0.9)$  from left to right. \label{conf_dist_ideal_circ_geod:fig}}
	\end{figure}

%
%
%
%
%
%

	Summarizing, this study indicates that tree boles of young Douglas fir trees with little competition grow uniformly or even stronger towards cylinders. Older tree boles their crowns competing for light, grow as if they keep ellipticality and tapering constant. These findings 
	certainly call for more elaborate research, e.g. tapering and ellipticality can be investigated by developing a method to study distance and projection to the submanifold of elliptical-like frusta.

\section{Conclusion and Outlook}

	In this paper for the statistical analysis of 3D shape, an asymptotic result for mean shape has been derived, a classical perturbation model and a newly proposed perturbation model for the statistical analysis of diffusion tensors has been revisited, and inference on the temporal deviation of the 3D shape of tree boles from cylinders has been performed.

	Although Procrustes means are very popular, it seems that due to a misunderstanding, the notion of 3D \emph{Procrustes population means} had not been quite available in the community. Instead, population means of a perturbation model 
	had been estimated by \emph{Procrustes sample means} which were, as was well known, ``consistent'' estimators under isotropic 2D errors. Introducing Fr\'echet $\rho$-means in this work and extending the available  Central-Limit Theorem 
	to underlying non-metrical distances, in particular, the issue of ``consistency'' has been identified as an ``incompatibility'' of the perturbation model with the shape space's geometry. Moreover, we recalled a Fr\'echet $\rho$-mean introduced by \cite{Z94} 
	the computation of which 
	in 3D is computationally slightly less costly than the computation of Procrustes means. This result allows for one-sample tests for a population Procrustes or Ziezold mean on the manifold part of the shape space, since due to strong consistency, sample Procrustes or Ziezold means will eventually come to lie on the manifold part a.s. For a two-sample test, one would need to ensure that Procrustes and Ziezold means are contained on the manifold part, if the underlying random shapes are a.s. contained in the manifold part. Settling this issue is the subject of a separate research, cf. \cite{H_meansmeans_10}. 

	While many means, different from classical intrinsic, extrinsic or Procrustes means are Fr\'echet $\rho$-means (e.g. geometric medians of \cite{fletch9} and (penalized) weighted Procrustes means of \cite{DKZL09}) it would be interesting to verify whether  a larger class of means, e.g. the semi-metrical mean introduced in \cite{STCB07b} and successfully employed in computer vision also fall into this setup. 

	Within the discussion of ``consistency'', the modeling of \cite{DKZ09} for diffusion tensor imaging has been extended 
	to diffusion tensors mildly deficient in rank. To the knowledge of the author, this is the first framework allowing for statistical inference on non-regular diffusion tensors without utilizing a Euclidean embedding, say, in the space of symmetric matrices. Of course in many applications, the interest lies specifically in degenerate tensors because these indicate a strong directional flow. The finding that mildly rank-deficient diffusion tensors are only ``mildly perturbation inconsistent'' may  provide for another motivation for the approach of \cite{DKZ09,DKZL09}. Following \cite{DKZ09}, we have used Procrustes means. For the underlying reflection shape space, extrinsic Schoenberg means qualify as well, they compute much faster than Procrustes or Ziezold means. As the most important advantage, the two-sample tests of \cite{B08} can then be performed even including rank 1 diffusion tensors. A drawback, however, may result from a possible insensitivity of Schoenberg means to degeneracy thus yielding a smaller discrimination power (cf. \cite{H_meansmeans_10} for a detailed discussion). These issues certainly warrant further research.

	Briefly compiling the findings on tree boles we can add to the well known fact that 
	longest boles can be found in dense stands, 
	it seems that the most cylindrical boles can be found in young stands with low competition.
	~\\~\\ \noindent{\bf
	A geodesic perturbation model.}
	As demonstrated in this work, for most practical applications in 3D avoiding degenerate configurations, perturbation models may be considered compatible with the shape space's geometry for isotropic error. Note that the shapes of entire tree boles are nearly singular (cf. \cite{HHM07}). For short bole frusta, however, such models may still serve  intuition and for sufficiently small error provide an adequate approximation. As one of such consider the \emph{geodesic perturbation model}
	\begin{eqnarray}\label{geod_perturbation_model:eq}
	x_i &=& g_i\lambda_i(\gamma(s_i) + \epsilon_i)
	\end{eqnarray}
	with a straight line $s \to \gamma(s) = \mu + s\nu$ in configuration space $F_m^k$ such that locally all points on $\gamma$ are in optimal  position to each other. 
	For 3D frusta, 
	the geodesic of cylinders has been used here, any other geodesic of specific frusta can be used similarly. Moreover, based on larger samples, a geodesic perturbation model within the space of elliptical-like frusta can be used to assess specific model parameters. In particular, a close investigation of such model parameters may lead to inference on the impact of environmental effects on the shape of tree boles building on artificially induced modifications of biological tree parameters as in \cite{gaslob04}.


\vspace{0.4cm}
\noindent\textbf{Acknowledgments:} 
The author would like to express his sincerest gratitude to the late Herbert Ziezold (1942 -- 2008) whose untimely death was most unfortunate. He would also like to thank Dieter Gaffrey and the colleagues from the Institute for Forest Biometry and Informatics at the University of G\"ottingen for discussion of and supplying with the tree-bole data. Also he is grateful for helpful discussions with Thomas Hotz, John Kent and Axel Munk. 

\appendix
\section{Appendix: Proofs}

 	\paragraph{The Central-Limit-Theorem for Fr\'echet $\rho$-means.}
%
	With the notation of Section \ref{Frech_means:scn}, suppose that $\rho: Q \times P\to \mathbb [0,\infty)$ is a map,  $\rho^2$ smooth in the second component as specified below. 
	In a local chart $(\phi,U)$ of $P$ near $\phi^{-1}(0)\in P$, denote by $\grad_2 \rho(q,p)^2$ the gradient of $x \mapsto \rho(q,\phi^{-1}(x))^2$ and by $H_2\rho(q,p)^2$ the corresponding Hessian matrix of second order derivatives. Then similar to \citet[p.1230]{BP05} consider the following integrability condition on a random variable $X$ on $Q$ at a location $\mu\in P$:
	\begin{eqnarray}\label{mean_integrability:cond}
	\left.\begin{array}{rcl}
	\mathbb E(\grad_2\rho(X,\mu)^2)&&\mbox{exists,}\\	 
	\cov(\grad_2\rho(X,\mu)^2)&&\mbox{exists,}\\	 
	\mathbb E(H_2\rho(X,\nu)^2)&&\mbox{exists for $\nu$ near $\mu$ and is continuous at $\nu=\mu$.}
	\end{array}
	\right\}
	\end{eqnarray}
	In analogy to condition (\ref{bdd_away:cond}) consider 
	\begin{eqnarray}\label{bdd_away_general:cond}
	\exists \epsilon >0 &\st& x \mapsto \rho(X,\phi^{-1}(x))^2 \mbox{ is smooth for } |x|<\epsilon  \mbox{ a.s.}\,.
	\end{eqnarray}
	The validity of (\ref{mean_integrability:cond}) and (\ref{bdd_away_general:cond}) is independent of the particular chart chosen. If $\rho$ is the Euclidean distance or $\rho = d_{S\Sigma_m^k}^{(i)}$, then (\ref{mean_integrability:cond}) is valid if $\mathbb E(\|X\|^2)<\infty$.

	If there is a discrete group $H$ acting on $P$ and $\mu \in P$ such that $E^{(\rho)}= \{h\mu:h\in H\}$, we  
	say that the Fr\'echet population $\rho$-mean set 
	is \emph{unique up to the action of $H$}. 
	
	\begin{Th}\label{CLT:th} 
	Suppose that $\mu$ is a point in the Fr\'echet $\rho$-mean set unique up to the action of $H$ on a manifold $P$ with respect to a continuous function $\rho:Q\times P\to \mathbb R$, $\rho^2$ smooth in the second component in the sense of (\ref{bdd_away_general:cond}), satisfying strong consistency, where $Q$ is a topological space. 
	If 
	\begin{enumerate}
	 \item[] for any measurable choice $p_n \in E^{(\rho)}_n$ there is a sequence $h_n\in H$ such that $\mu_n=h_np_n \to \mu $ a.s., and if
	\begin{enumerate}
	\item[(i)] $X$ has compact support or
	\item[(ii)] the integrability conditions (\ref{mean_integrability:cond}) are satisfied at $\mu$ 
	\end{enumerate}
	\end{enumerate}
	then for any measurable choice $p_n \in E^{(\rho)}_n$ there is a sequence $h_n\in H$ such that $\mu_n=h_np_n$ satisfies a CLT.
	In a suitable chart $(\phi,U)$ the corresponding matrices from Definition \ref{CLT:def} are given by
	$$ A = \mathbb E(H_2\rho(X,\mu)^2),~~\Sigma = \cov(\grad_2\rho(X,\mu)^2)\,.$$
	\end{Th}

	\begin{proof} Obviously, if $X$ has compact support then (\ref{mean_integrability:cond}) is satisfied. Hence we may assume the case (ii).  We adapt the ideas laid out in \citet[p. 1229--1230]{BP05}. Let $D$ denote the dimension of $P$ and consider a local chart $(U,\phi)$ near $\mu=\phi^{-1}(0)$. 
	We have a.s. eventually that $\mu_n\in U$, then $\phi(\mu_n) =x_n \to 0$ a.s. Abbreviating 
	$g(x) = \sum_{j=1}^n\rho(X_j,\phi^{-1}(x))^2$, $\grad g(x) = (g_1(x),\ldots,g_D(x))^T$ and $Hg(x)$ for the Hessian 
	we have by (\ref{bdd_away_general:cond}) the Taylor expansion
	{\footnotesize \begin{eqnarray*}
	0&=& \frac{1}{\sqrt{n}}\grad g(x_n)~=~ \frac{1}{\sqrt{n}}\grad g(0) + \frac{1}{\sqrt{n}} \left(\begin{array}{c}
	(\grad g_1(t_1x_n))^T\cdot x_n\\
	\vdots\\
	(\grad g_D(t_Dx_n))^T\cdot x_n
	 \end{array}\right)\\
	&=& \frac{1}{\sqrt{n}}\grad g(0) + \frac{1}{n}Hg(0)\cdot \sqrt{n} x_n+ \frac{1}{n}\left(\begin{array}{c}
	(\grad g_1(t_1x_n)-\grad g_1(0))^T\cdot \sqrt{n}x_n\\
	\vdots\\
	(\grad g_D(t_Dx_n)-\grad g_D(0))^T\cdot \sqrt{n}x_n
	 \end{array}\right)
	\end{eqnarray*}}\noindent
	for suitable $0\leq t_1,\ldots,t_D\leq 1$ a.s.
	In conjunction with the classical CLT, the first two conditions in (\ref{mean_integrability:cond}) ensure that
	{\footnotesize$$ \frac{1}{\sqrt{n}}\grad g(0)~=~\frac{1}{\sqrt{n}}\sum_{j=1}^n  \grad_2 \rho(X_j,\phi^{-1}(0))^2 ~\to~ -{\cal G}$$}in distribution for some Gaussian matrix ${\cal G}$ with mean $\mathbb E\big(\grad_2 \rho(X,\phi^{-1}(0))^2\big) =0$ and covariance matrix $\cov\big(\grad_2 \rho(X,\phi^{-1}(0))^2\big)$. The third condition guarantees the Strong Law of Large Numbers for the mean of random Hessians 
	{\footnotesize$$ \frac{1}{n}Hg(0)~=~\frac{1}{n} \sum_{j=1}^n  H_2 \rho(X_j,\phi^{-1}(0))^2 ~\to~ \mathbb E(H_2\rho(X,\mu)^2)=:A~~a.s.$$}Note that the $k$-th entry of $ \frac{1}{n}\Big(\grad g_i(t_ix_n)-\grad g_i(0)\Big) $ for
	$1\leq i,k\leq D$ is
	{\footnotesize$$ \frac{1}{n}\sum_{j=1}^n\left(\frac{\partial^2}{\partial_i\partial_k}\rho(X_j,\phi^{-1}(t_ix_n))^2 -\frac{\partial^2}{\partial_i\partial_k}\rho(X_j,\phi^{-1}(0))^2\right)$$}which 
	tends to zero a.s. again by the third condition of (\ref{mean_integrability:cond}).
	In consequence
	$A \sqrt{n} x_n ~\to~{\cal G}$
	in distribution yielding the assertion.	
	\end{proof}
%
%
	\noindent{\it Proof of Theorem \ref{diff_tens_kend_ss:thm}.}
	First in I we define an extension of the Cholesky factorization which yields that $\tau$ and $\tau_s$ are well defined. Then in II we show the topological properties of $\tau$ and $\tau_s$. The other assertions of Theorem \ref{diff_tens_kend_ss:thm} are then straightforward. 

	I: Define an extension of the Cholesky factorization by $P(m)\ni a=v\lambda^2 v^T \to b=\sqrt{\lambda}v^T = gu \to u\in UT^\geq(m)$ where  $\lambda$ is a diagonal matrix with non-negative entries, $v,g\in SO(m)$. Then, obviously $u^Tu=a$. The factorization $b = gu$ is obtained by the usual Gram-Schmidt process, i.e. the first column $g_1$ of $g$ is $b_{(1)}/\|b_{(1)}\|$ where $b_{(1)}$ denotes the first non-zero column in $b$, the second column is $g_2= (b_{(2)} - g_1 b_{(2)}^Tg_1)/\|b_{(2)} - g_1 b_{(2)}^Tg_1\|$ with $b_{(2)}$, the first column in $b$ linearly independent of $g_1$, etc.. Either proceed until $g_m$ (note that we have $u_{mm}\geq 0$ in consequence of $\det(b)\geq 0$ and $g\in SO(m)$) or extend with arbitrary columns such that $g\in SO(m)$. In order to see that this mapping is well defined suppose that $u'\in UT^\geq(m)$ and $g'\in SO(m)$ with
	$u'^Tu' = uu^T$. Set
	$$u = \left(\begin{array}{c}q\\\hline 0\end{array}\right),~~u' = \left(\begin{array}{c}q'\\\hline 0\end{array}\right)$$
	with a matrices $q,q' \in M(r,m)$ of rank $r$. Then, there are indeces $1\leq j_1\leq \ldots\leq j_r$ such that
	the matrix $\widetilde{q}=(q_{j_1},\ldots,q_{j_r})$ composed of linearly independent columns $q_{j_1},\ldots,q_{j_r}$ of $q$ is upper triangular and regular, i.e. $ \widetilde{q}\in UT^+(r)$. Denoting by 
	$\widetilde{q}'=(q'_{j_1},\ldots,q'_{j_r})$ the matrix composed of the corresponding columns of $q'$ note that by construction $\widetilde{q}'\in UT(r)$. By hypothesis, $\widetilde{q}'^T\widetilde{q}'=\widetilde{q}^T\widetilde{q}$ which yields with the classical Cholesky decomposition that $\widetilde{q}'=\widetilde{q}$. Since $\widetilde{q}'^T\hat{q}' = \widetilde{q}^T\hat{q}$ with ``$\hat{\cdot}$'' standing for the complementary columns in $\{1,\ldots,r\}$ we have indeed $q'=q$, i.e. $u'=u$. Obviously this mapping is an extension of the classical Cholesky factorization.
	
	The bijective mappings $\phi :\Sigma_{m}^{m+1}\to SUT(m)$ and $\phi_s :S\Sigma_{m}^{m+1}\to UT(m)$ are similarly obtained (cf. \citet[Section 1.3]{KBCL99}) by a Gram-Schmidt decomposition of $a \in [a] \in S\Sigma_{m}^{m+1}$. To this end choose the sign of $u_{mm}$ such that $g\in SO(m)$ which possibly gives a negative entry $u_{mm}$. Note that 
	\begin{eqnarray*}
	\phi_s\circ \tau_s\big(P(n)\big) = UT^{\geq}(m)&\mbox{and}& \phi\circ \tau\big(P(n)\big) = SUT^{\geq}(m)\,.
	\end{eqnarray*}
	This shows that $\tau$ and $\tau_s$ are indeed well defined. Moreover, $\tau_s$ is injective.

	II: We now rely on the decomposition of $UT(m)$ and $SUT(m)$ as simplicial complexes by \citet[Chapter 2]{KBCL99} defined through the topologies of $S\Sigma_m^{m+1}$ and $\Sigma_m^{m+1}$, respectively. Verify that the boundary identification on the respective subsets $UT^\geq(m)$ and $SUT^\geq(m)$ is the same as given by the natural topology of $P(m)$. Hence,  $\tau$ and $\tau_s$ are both open and continuous.
	\qed

	\begin{Lem}\label{chol:lem} Let $e\in F_m^{m+1}$ be the unit matrix and $a=(e + \epsilon)^T(e + \epsilon)$ with error $\epsilon\in F_m^{m+1}$ independent standard normal in every component. Then, 
	$$\mathbb E\big({\rm chol}(a)\big)\neq e\,.$$\end{Lem}
	\begin{proof}
	 Let $e +\delta ={\rm chol}(a)$ giving for the entry in the first column and row 
	{\footnotesize 
	$$1+\delta_{11} = \sqrt{1+2\epsilon_{11} + \sum_{j=1}^m\epsilon_{j1}^2}\,.$$}Since $\sqrt{(1+x)^2+y^2} + \sqrt{(1-x)^2+y^2}\geq 2\sqrt{x^2+y^2}$ observe that 
	$$ \mathbb E\left(\sqrt{(1+\epsilon_{11})^2+\sum_{j=2}^m\epsilon_{j1}^2}\right) \geq  \mathbb E\left(\sqrt{\sum_{j=1}^m\epsilon_{j1}^2}\right) =\sqrt{2}\,\frac{\Gamma\left(\frac{m+1}{2}\right)}{\Gamma\left(\frac{m}{2}\right)} ~>~ 1$$
	for $m>1$.
	Hence $\mathbb E(e+\delta) =e$ cannot be. 
	\end{proof}
%
%
%
	\noindent{\it Proof of Theorem \ref{parallel_cone_scn_geodesic:thm}.}
	 Since the diagonal of $w_{\alpha,\beta,r,t}w_{\alpha',\beta',r',t'}^T$ consists of positive values only, all elements of $PS_{{\bf a},{\bf b}}$ are mutually unique in optimal position (\citet[p. 114]{KBCL99}). Since straight lines are mapped under the Helmert sub-matrix to straight lines in the configuration space which project to great circles in the pre-shape space, the straight line segment between $w_{\alpha,\beta,r,t}$ and $w_{\alpha',\beta',r',t'}$ maps to a segment on a horizontal geodesic on $S_m^k$. In consequence, the projection of $PS_{{\bf a},{\bf b}}$ to $(\Sigma_3^{2\kappa})^*$ is a totally geodesic submanifold, its horizontal lift  is again $PS_{{\bf a},{\bf b}}$, cf. \citet[Section 2]{HHM07}.
	\qed 

\bibliographystyle{elsart-harv}
\bibliography{shape,botany}
\end{document}